\documentclass[12pt,reqno]{amsart}

\usepackage{graphicx}%
\usepackage{multirow}%
\usepackage{amsmath,amssymb,amsfonts}%
\usepackage{amsthm}%
\usepackage{mathrsfs}%
\usepackage[title]{appendix}%
\usepackage{xcolor}%
\usepackage{textcomp}%
\usepackage{manyfoot}%
\usepackage{booktabs}%
\usepackage{algorithm}%
\usepackage{algorithmicx}%
\usepackage[noend]{algpseudocode}%
\usepackage{listings}%
\usepackage{url}

\usepackage{tikz}
\usepackage{colortbl}
\usepackage{complexity}
\usepackage{caption}
\usepackage{hyperref}
\usepackage{wrapfig}
\usepackage[margin=1in]{geometry}

\newcommand\algoIndent[1]{%
  \linebreak\hspace*{\dimexpr\algorithmicindent*#1}%
}

\newcommand{\imm}{\mathrm{im}}

\newcommand{\len}{\textsf{len}}
\newcommand{\rank}{\mathrm{rank}}

\newtheorem{theorem}{Theorem}[section]
\newtheorem{lemma}[theorem]{Lemma}

\theoremstyle{definition}
\newtheorem{definition}{Definition}[section]
\newtheorem{example}{Example}[section]

\title[Computing eulerian magnitude homology] {Computing eulerian magnitude homology}
\author{Giuliamaria Menara}
\address[GM]{Department of Mathematics, Informatics and Geosciences, University of Trieste, Trieste, 34127, Italy}
\email{giuliamaria.menara@phd.units.it}

\author{Luca Manzoni}
\address[LM]{Department of Mathematics, Informatics and Geosciences, University of Trieste, Trieste, 34127, Italy}
\email{lmanzoni@units.it}


\begin{document}
\maketitle

\begin{abstract}
    In this paper tackle the problem of computing the ranks of certain eulerian magnitude homology groups of a graph $G$. First, we analyze the computational cost of our problem and prove that it is $\#W[1]$-complete. Then we develop the \emph{first diagonal algorithm}, a breadth-first-search-based algorithm parameterized by the diameter of the graph to calculate the ranks of the homology groups of interest. To do this, we leverage the close relationship between the combinatorics of the homology boundary map and the substructures appearing in the graph. We then discuss the feasibility of the presented algorithm and consider future perspectives.
\end{abstract}

\section{Introduction}
\label{sec:intro}

Graphs are used in a large variety of fields to model and analyze complex relationships. 
In this framework, the search for interesting and relevant substructures is a standard procedure, and the detection of cliques and clique-like subgraphs is a fundamental tool in graph analysis.
Such substructures have been applied in many different situations including community detection in social networks~\cite{kumar1999trawling,sozio2010community}, identification of real-time stories in the news~\cite{angel2014dense} and graph visualization~\cite{zhang2012extracting,zhao2012large}.
In practice, due to noise in data, one is also interested in large ``near-cliques''. 
While this is not a standard term, applications involve cliques that lack a small sparse subgraph.
For example, incomplete cliques have been used to predict missing pairwise interactions~\cite{zhang2012extracting} and for identifying functional groups~\cite{han2007identifying} in a protein-protein interaction network. 
Also, they were exploited for community detection~\cite{zhu2020community} and for detecting test collusion~\cite{belov2021graph}.
Recent works have used the fraction of near-cliques to $k$-cliques to define higher order variants of clustering coefficients~\cite{yin2017local}.

In the present work, in order to quantitatively characterize these structures, we propose an efficient algorithm to compute eulerian magnitude homology, a tool that comes from the field of algebraic topology.
Magnitude is an isometric invariant of metric spaces, so-named for its web of connections to ``size-like'' quantities of significance in various corners of mathematics.
Defined and first studied by Leinster~\cite{leinster2013magnitude}, it is a special case of a general theory of magnitude of an enriched category, and has found applications in areas like biodiversity (e.g., Leinster and Cobbold~\cite{leinster2012measuring}).
As a finite graph naturally gives rise to a finite metric space, it is possible to associate magnitude with it.
Magnitude homology has been invented by Hepworth and Willerton~\cite{hepworth2015categorifying} as an enrichment of the magnitude of a graph which is equipped with a graph metric.
The magnitude homology of graphs has been well studied in recent years and has proven to be a rich invariant~\cite{gu2018graph,hepworth2022magnitude,hepworth2015categorifying,kaneta2021magnitude,sazdanovic2021torsion}.
Also, several theoretical tools for computing the magnitude homology of a graph have been studied so far. 
For example, Hepworth and Willerton proved in~\cite{hepworth2015categorifying} a Mayer-Vietoris type exact sequence and a Kunneth type formula, and Gu~\cite{gu2018graph} uses algebraic Morse theory for computation for some graphs. 
\emph{Eulerian magnitude homology} is a variant of magnitude homology defined by Giusti and Menara in~\cite{giusti2024eulerian} with the specific purpose of exploiting the relation between magnitude homology and graph substructures.
Although the authors were able to provide a meaningful interpretation of the first diagonal eulerian magnitude homology groups of a graph $G$ (that is, groups $EMH_{k,\ell}(G)$ with $\ell = k$), a major drawback is represented by the fact that computing eulerian magnitude homology requires the enumeration of $k$-simple paths of $G$ and, as the number of vertices $n$ grows, the number $K$ of $(k+1)$-tuples inducing a path of length $k$ becomes exponential in the number $k+1$ of required vertices (i.e. $K \sim n^{k+1}$).

It was shown in \cite{giusti2024eulerian} that searching for all $(k+1)$-tuples inducing a path of length $k$ and generating homology is the same as looking for specific substructures in the graph $G$.
In other words, computing the rank of eulerian magnitude homology groups $EMH_{k,k}(G)$ is the same of enumerating subgraphs of $G$ belonging to a certain family $\mathcal{H}$.
This problem is called the \emph{subgraph isomorphism problem} and there is an extensive literature studying its computational complexity.
For example, it is shown in~\cite{hell1990complexity} that for any fixed simple graph $G$, the problem of whether there exists an isomorphism from another graph $H$ to $G$ is solvable in polynomial time if $G$ is bipartite, and \NP-complete if $G$ is not bipartite.
Also, the work by Dyer and Greenhill~\cite{dyer2000complexity} proves that polynomial-time solvable cases arise only when $G$ is an isolated vertex, a complete graph with all loops present, a complete bipartite graph without loops, or a disjoint union of these graphs.
Moreover, in~\cite{amini2012counting} Amini, Fomin and Saurabh relate counting subgraphs to counting graph isomorphisms.
They provide exact algorithms for several problems (compute the number of optimal bandwidth permutations of a graph on $n$ vertices excluding a fixed graph as a minor, counting all subgraphs excluding a fixed graph $M$ as a minor, counting all subtrees with a given maximum degree) and all of them can be solved in time at least $2^{\mathcal{O}(n)}$.

In this paper, we will prove the intrinsic difficulty of our problem by showing that it is complete for the $\#W[1]$ complexity class.
Then, we tackle this computational problem and propose a breadth-first-search-based approach to compute the eulerian magnitude homology groups $EMH_{k,k}(G)$, which results in an algorithm that is more computationally efficient than relying directly on the definition.
Indeed, even if in the worst-case scenario we still have exponential computational complexity, for many graphs emerging from real-world scenarios the complexity is sub-exponential or even polynomial, as we will show in Section \ref{subsec:discussion}.

\subsection{Outline}
The paper is organized as follows.
We start by recalling in Section \ref{sec:background} some general background about graphs and eulerian magnitude homology.
Then in Section \ref{sec:EMH computational cost} we investigate the complexity of computing eulerian magnitude homology groups and we tie this problem to the subgraph isomorphism problem.
In Section \ref{sec:algorithm} we develop and thoroughly describe an algorithm to compute first diagonal eulerian magnitude homology groups of a graph $G$. 

\subsection{List of abbreviations}
The following abbreviations will be used throughout the paper.

\begin{itemize}
    \item Breadth First Search: BFS
    \item Eulerian magnitude chain: EMC
    \item Eulerian magnitude homology: EMH
    \item First Diagonal Algorithm: FDA
    \item Magnitude chain: MC
    \item Magnitude homology: MH
\end{itemize}

\section{Background}
\label{sec:background}

We begin by recalling relevant definitions and results.
We assume readers are familiar with the general theory of simplicial homology (for a thorough exposition see~\cite{hatcher2005algebraic}).
Throughout the paper we adopt the notation $[m] = \{1, \dots, m\}$ and $[m]_0 = \{0, \dots, m\}$ for common indexing sets.

\subsection{Graph terminology and notation}
\label{subsec:graph definitions}

An undirected graph is a pair $G=(V,E)$ where $V$ is a set of vertices and $E$ is a set of edges (unordered pairs of vertices).
Recall that a \emph{walk} in $G$ is an ordered sequence of vertices $x_0,x_1,\ldots,x_k \in V$ such that there is an edge $\{x_i,x_{i+1}\}$ for all $i \in [k]_0$.
We call a \emph{path} a walk with no repeated vertices.
For the purposes of introducing eulerian magnitude homology, we assume all graphs to have no self-loops and no multiedges~\cite{leinster2019magnitude}.
It is possible to consider the set of vertices of a graph as an extended metric space (i.e. a metric space with infinity allowed as a distance) by taking the \emph{hop  distance} $d(u,v)$ to be equal to the length of a shortest path in $G$ from $u$ to $v$, if such a path exists, and taking $d(u,v) = \infty$ if $u$ and $v$ lie in disconnected components of $G$.

\begin{definition}
\label{def:ktrail}
Let $G = (V,E)$ be a graph, and $k$ a non-negative integer. A \emph{$k$-trail} $\bar{x}$ in $G$ is a $(k+1)$-tuple $(x_0,\dots,x_k) \in V^{k+1}$ of vertices for which $x_i \neq x_{i+1}$ and $d(x_i,x_{i+1})<\infty$ for every $i \in [k-1]_0$.
The \emph{length} of a $k$-trail $(x_0,\dots,x_k)$ in $G$ is defined as the minimum length of a walk that visits $x_0,x_1,\ldots,x_k$ in this order:
\[
\len(x_0,\dots,x_k) = d(x_0,x_1)+\cdots + d(x_{k-1},x_k).
\]
We call the vertices $x_0, \dots x_{k}$ the \emph{landmarks}, $x_0$ the \emph{starting point}, and $x_k$ the \emph{ending point} of the $k$-trail.
\end{definition}

\subsection{Eulerian magnitude homology}
\label{subsec:EMH definitions}

The \emph{magnitude homology} of a graph $G$ was first introduced by Hepworth and Willerton in~\cite{hepworth2015categorifying}, and the \emph{eulerian} magnitude homology of a graph is a variant of it with a stronger connection to the subgraph structures of $G$.
Specifically, while the building blocks of standard magnitude homology are tuples of vertices $(x_0,\dots,x_k)$ where we ask that \emph{consecutive} vertices are different, eulerian magnitude homology is defined starting from tuples of vertices $(x_0,\dots,x_k)$ where we ask that \emph{all} landmarks are different.

Eulerian magnitude homology was recently introduced by Giusti and Menara in~\cite{giusti2024eulerian} and we recall here the construction.

\begin{definition}(Eulerian magnitude chain)
\label{def:magchain}
	Let $G=(V,E)$ be a graph.
	We define the $(k,\ell)$-eulerian magnitude chain $EMC_{k,\ell}(G)$ to be the free abelian group generated by trails $(x_0,\dots,x_k) \in V^{k+1}$ such that $x_i \neq x_j$ for every $0\leq i,j \leq k$ and $\len(x_0,\dots,x_k)=\ell$.
\end{definition}

It is straightforward to demonstrate that the eulerian magnitude chain is trivial when the length of the path is too short to support the necessary landmarks.

\begin{lemma}[{c.f. \cite[Proposition 10]{hepworth2015categorifying}}]
\label{lem:LowerTriangular}
Let $G$ be a graph, and $k > \ell$ non-negative integers. Then $EMC_{k,\ell}(G) \cong 0.$
\end{lemma}

\begin{proof}
    Suppose $EMC_{k,\ell}(G)\neq 0.$
    Then, there must exist a $k$-trail $(x_0,\dots,x_k)$ in $G$ so that $\len(x_0,\dots,x_k)=d(x_0,x_1)+\cdots+d(x_{k-1},x_k) = \ell$.
	However, as all vertices in the $k$-trail must be distinct, $d(x_i,x_{i+1}) \geq 1$ for $i \in [k-1]_0$, so $k$ can be at most $\ell$.
\end{proof}

\begin{definition}(Differential)
	\label{def:differential}
	Denote by $(x_0,\dots,\hat{x_i},\dots,x_k)$ the $k$-tuple obtained by removing the $i$-th vertex from the $(k+1)$-tuple $(x_0,\dots,x_k)$.  We define the \emph{differential}
	\[
	\partial_{k,\ell}: EMC_{k,\ell}(G) \to EMC_{k-1,\ell}(G)
	\]
	to be the signed sum $\partial_{k,\ell}= \sum_{i\in [k-1]}(-1)^{i}\partial_{k,\ell}^i$ of chains corresponding to omitting landmarks without shortening the walk or changing its starting or ending points,
	\[
	\partial_{k,\ell}^i(x_0,\dots,x_k) = \begin{cases}
		(x_0,\dots,\hat{x_i},\dots,x_k) , &\text{ if } \len(x_0,\dots,\hat{x_i},\dots,x_k) = \ell, \\
		0, &\text{ otherwise.}\\
	\end{cases}
	\]
\end{definition}

For a non-negative integer $\ell$, we obtain the \emph{eulerian magnitude chain complex}, $EMC_{*,\ell}(G),$ given by the following sequence of free abelian groups and differentials.

\begin{definition}(Eulerian magnitude chain complex)
\sloppy
We indicate as $EMC_{*,\ell}(G)$ the following sequence of free abelian groups connected by differentials
	\[
	\cdots \rightarrow EMC_{k+1,\ell}(G) \xrightarrow{\partial_{k+1,\ell}} EMC_{k,\ell}(G) \xrightarrow{\partial_{k,\ell}} EMC_{k-1,\ell}(G) \to \cdots
	\]
\end{definition}

The differential map used here is the one induced by standard magnitude, and it is shown in \cite[Lemma 11]{hepworth2015categorifying} that the composition $\partial_{k,\ell} \circ \partial_{k+1,\ell}$ vanishes, justifying the name ``differential" and allowing the definition the corresponding bigraded homology groups of a graph.

\begin{definition}(Eulerian magnitude homology)
	\label{def_EMH}
	The $(k,\ell)$-eulerian magnitude homology group of a graph $G$ is defined by
	\[
	EMH_{k,l}(G) = H_k(EMC_{*,l}(G)) = \frac{\ker(\partial_{k,\ell})}{\imm(\partial_{k+1,\ell})}.
	\]
\end{definition}

\begin{example}
\label{ex:toyexampleEMH}
\sloppy
Consider the graph $G$ in Figure \ref{fig:toyexampleEMH}.
We will compute $EMH_{2,2}(G)$.
$EMC_{2,2}(G)$ is generated by the $2$-paths in $G$ of length $2$.
There are ten such paths, consisting of all possible walks of length two in the graph visiting different landmarks: (0,1,2), (0,2,1), (0,2,3), (1,0,2), (1,2,0), (1,2,3), (2,0,1), (2,1,0), (3,2,0), (3,2,1).
Similarly, $EMC_{1,2}(G)$ is generated by the four $1$-paths in $G$ of length $2$: (0,3), (1,3), (3,0), (3,1).
In this case the differential $\partial_{2,2}$ only omits the center vertex.
Thus it is straightforward to check that it is surjective, and that the kernel is generated by the $6$ elements whose length diminishes when the middle vertex is removed; that is, all elements in $EMC_{2,2}(G)$ except (0,2,3), (1,2,3), (3,2,0), (3,2,1).  
On the other hand, by Lemma \ref{lem:LowerTriangular}, $EMC_{3,2}(G)$ is the trivial group, and thus the image of $\partial_{3,2}$ is $\langle 0 \rangle$. Therefore $\rank(EMH_{2,2}(G))=6$ and $EMH_{2,2}(G)$ is generated by those walks between vertices 0, 1, and 2 of the triangle, recording the fact that there is a shorter path between the starting and ending points of the walks.
\end{example}

\begin{figure}[ht]	
\centering
\begin{tikzpicture}[node distance={15mm}, thick, main/.style = {draw, circle}]
    \node[main] (0) {$0$}; 
    \node[main] (1) [right of=0] {$1$};  
    \node[main] (2) [above of=1] {$2$};
    \node[main] (3) [right of=1] {$3$};
    \draw (0) -- (1);
    \draw (1) -- (2);
    \draw (0) -- (2);
    \draw (2) -- (3);
\end{tikzpicture} 
\caption{Graph $G$}
\label{fig:toyexampleEMH}
\end{figure}

\section{Eulerian magnitude homology computational cost}
\label{sec:EMH computational cost}

Relying on Definition \ref{def_EMH} to compute eulerian magnitude homology results in a computationally unfeasible task.
Indeed, consider for example the $(k,k)$-EMH group.
The sole construction of the eulerian magnitude chain $EMC_{k,k}(G)$ requires $k!$ checks for every possible sequence of $(k+1)$ vertices, for a total of $\binom{n}{k+1} k! \sim n^{k+1}$ checks.
Therefore, in the worst-case scenario of a complete graph $G=K_{n+1}$, the last computable chain $EMC_{n,n}(G)$ requires $n^{n+1}$ computations.

One possible approach to this issue is turn it into a subgraph counting problem.
Indeed, it was proven in \cite{giusti2024eulerian} (although using a different language) that homology cycles in the eulerian magnitude chain complex can be decomposed into cycles supported on specific subgraphs $H_i$ belonging to a family $\mathcal{H}=\{H_i\}_i$. 
In what follows, we recall a useful result from \cite{giusti2024eulerian} and then we proceed by explicitly exhibiting the graphs family $\mathcal{H}$.

\begin{lemma}[\cite{giusti2024eulerian}]
    \label{lem:diff_zero_so_edge}
    Let $G = (V, E)$ be a graph and take $k \geq 2.$ Fix some $i \in [k-1]$ and $\bar{x} = (x_0, x_1, \dots, x_k)\in EMC_{k,k}(G).$
    Then $\partial^i_{k,\ell}(\bar{x}) = 0$ if and only if $\{x_{i-1}, x_{i+1}\}\in E.$     
\end{lemma}

\begin{proof}
    Let $G = (V, E)$ be a graph.
    Take $k \geq 2,$ $i \in [k-1],$ and $\bar{x} = (x_0, x_1, \dots, x_k)\in EMC_{k,k}(G).$
    As $\len((x_0, \dots, x_k)) = k$ and each vertex is distinct, we have $d(x_j, x_{j+1}) = 1,$ so $\{x_j, x_{j+1}\}\in E$ for each $j \in [k-1]_0.$
    As $\partial^i_{k,k}(\bar{x}) \in EMC_{k-1,k}(G),$ if $\partial^i_{k,k}(\bar{x}) = 0,$ then $\len(\partial^i_{k,k}(\bar{x})) < \len(\bar{x}) = k,$ which can only occur if removing the landmark $x_i$ shortens the walk, which occurs precisely when $\{x_{i-1}, x_{i+1}\} \in E$.
    Thus the graph $H(\bar{x})=\left(V_{\bar{x}},E_{\bar{x}} \right)$ induced by a single homology generator is defined as
    
    \begin{align*}
        V_{\bar{x}} &= V(\bar{x}) \\
        E_{\bar{x}} &= \{\{x_i, x_{i+1}\}\}_{i\in[k-1]_0} \cup \{\{x_{i-1}, x_{i+1}\}\}_{i\in[k-1]}\,
    \end{align*}
    
    see Figure \ref{fig:single-tuple-subgraph-emh}.
\end{proof}

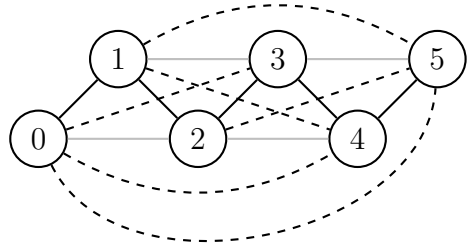
\begin{figure}[ht]
    \centering
    \begin{tikzpicture}[node distance={15mm}, thick, main/.style = {draw, circle,minimum size=.75cm}]
            \node[main] (1) {$1$};
			\node[main] (0) [below left of=1] {$0$};  
			\node[main] (2) [below right of=1] {$2$};
			\node[main] (3) [above right of=2] {$3$};
			\node[main] (4) [below right of=3] {$4$};
			\node[main] (5) [above right of=4] {$5$};
			\draw (0) -- (1);
			\draw (1) -- (2);
			\draw (2) -- (3);
			\draw (3) -- (4);
			\draw (4) -- (5);
			
			\draw[lightgray] (0) -- (2);
			\draw[lightgray] (1) -- (3);
			\draw[lightgray] (2) -- (4);
			\draw[lightgray] (3) -- (5);

			\draw[dashed] (0) -- (3);
			\draw[dashed] (1) -- (4);
            \draw[dashed] (2) -- (5);

           \path[every node/.style={font=\sffamily\small}]
			(0) edge [bend right, dashed] node {} (4);
            \path[every node/.style={font=\sffamily\small}]
			(0) edge [bend right=75, dashed] node {} (5);
            \path[every node/.style={font=\sffamily\small}]
			(1) edge [bend left, dashed] node {} (5);
		\end{tikzpicture}
		\caption{Subgraph $H(\bar{x})$ induced by $[\bar{x}]=[0,1,2,3,4,5] \in EMH_{5,5}(G)$. The edges in the path $(0,1,2,3,4,5)$ are represented in black. Since the removal of each vertex causes the length of the induced path to decrease, it means all grey edges $\{x_{i-1},x_{i+1}\}$ are contained in the induced graph. The dashed edges do not play a role in the homology computation.}
		\label{fig:single-tuple-subgraph-emh}   
\end{figure}

Generalizing Lemma \ref{lem:diff_zero_so_edge} to the situation where a linear combination of tuples $\sum_i a_i \bar{x}_i$ is an eulerian magnitude homology cycle becomes difficult because of the increasingly complicated collection of isomorphism types of graphs which can support eulerian trails.

A first result in this direction is the characterization of the subgraph $H$ induced by two tuples $\bar{x}^1$ and $\bar{x}^2$ such that $\partial_{k,k}(\bar{x}^1-\bar{x}^2)=0.$
We describe this case in the following example.

\begin{example}
    \label{ex:two_trail_cycle_graph}
    Let $G=(V, E)$ be a graph, and let $\bar{x}^i = (x_0^i, \dots, x_k^i), i=1, 2,$ be trails in $EMC_{k,k}(G).$
    We want to construct a graph $H$ induced by the vertex set $V(\bar{x}^1) \cup V(\bar{x}^2)$ such that the difference $(\bar{x}^1-\bar{x}^2)$ is a non-trivial generator of $EMH_{k,k}(G)$, i.e. $\bar{x}^1$ and $\bar{x}^2$ are generators of $EMC_{k,k}(G)$ for which  $\partial_{k,k}(\bar{x}^i) \neq 0, i=1, 2$ but $\partial_{k,k}(\bar{x}^1 - \bar{x}^2) = 0.$
    So, say there is one landmark $\bar{x}^i_r$ such that $\partial_{k,k}(\bar{x}^i)=(-1)^r\partial_{k,k}^r(\bar{x}^i) \neq 0$. 
    
    From the definition of the differential, if $\bar{x}^1 \neq \bar{x}^2$ and $\partial_{k,k}(\bar{x}^1 - \bar{x}^2) = 0$ then both trails agree in all landmarks except one, say $x_r^1 \neq x_r^2$ for some $r \in [k-1],$ and the vertex $x_r^2$ cannot appear as a landmark in $\bar{x}^1$ nor vice versa.
    Indeed, suppose they differ in two landmarks, say $x_r^1 \neq x_r^2$ and $x_s^1 \neq x_s^2$ for some $r,s \in [k-1]$, and say $r<s$.
    Then if we compute the boundary $\partial_{k,k}(\bar{x}^1 - \bar{x}^2)$ we get
    \begin{align*}
        \partial_{k,k}(\bar{x}^1 - \bar{x}^2) &= \sum_{i=1}^{k-1} (-1)^i \partial_{k,k}^i(\bar{x}^1 - \bar{x}^2) \\
        &= (-1)^r \partial_{k,k}^r(\bar{x}^1 - \bar{x}^2) + (-1)^s \partial_{k,k}^s(\bar{x}^1 - \bar{x}^2) \\
        &= (-1)^r ((x_0^1,\dots, \hat{x}_r^1, \dots,x_s^1, \dots, x_k^1 ) - (x_0^2,\dots, \hat{x}_r^2, \dots,x_s^2, \dots, x_k^2 )) \\
        &+ (-1)^r ((x_0^1,\dots, x_r^1, \dots,\hat{x}_s^1, \dots, x_k^1 ) - (x_0^2,\dots, x_r^2, \dots,\hat{x}_s^2, \dots, x_k^2 )) \\
        &\neq 0,
    \end{align*}

    because we assumed $x_r^1 \neq x_r^2$ and $x_s^1 \neq x_s^2$, and thus the subtuples do not simplify.
    In particular, $x_0^1 = x_0^2$ and $x_k^1 = x_k^2,$ so the trails have the same starting and ending points.
    
    Let $H(\bar{x}^1) = (V^{\{1\}}, E^{\{1\}})$ be the graph of Lemma \ref{lem:diff_zero_so_edge} represented in Figure \ref{fig:single-tuple-subgraph-emh}.
    The set $E^{\{1\}}$ necessarily contains all of the edges in the support of both trails except $\{x_{r-1}^1, x_r^2\}$ and $\{x_r^2, x_{r+1}^1.\}$
    Further, from the proof of Lemma \ref{lem:diff_zero_so_edge} we know that the edges already present in $H(\bar{x}^1)$ imply $\partial_{k,k}^i(\bar{x}^1) = 0, i \neq r,$ and $\partial_{k,k}^i(\bar{x}^2) = 0$ for $i \neq r-1, r, r+1,$ because the two trails are equal away from these vertices.
    However, we must introduce new edges to ensure $\partial_{k,k}^{r-1}(\bar{x}^2) = \partial_{k,k}^{r+1}(\bar{x}^2) = 0.$
    So, we define $H = (V_H, E_H),$ where 
    \begin{align*}
    V_H = &V^{\{1\}} \cup \{x_r^2\}\\
    E_H = &\left(E^{\{1\}} \cup \{\{x_{a}^1, x_r^2\} \;\colon\; a \in \{r-2, r-1, r+1, r+2\} \cap [k]_0\}\right) \\
    &\setminus \{ \{x_{r-1}^1, x_{r+1}^1\}\}.
     \end{align*}
   
    See Figure \ref{fig:subgraph emh m-combination} for an illustration.
    The new vertex $x_r^2$ and the two new edges $\{x_{r-1}^1, x_r^2\}$ and $\{x_{r+1}^1, x_r^2\}$ support the trail $\bar{x}^2$ and imply the agreement of $\partial_{k,k}^r$ on the two generators.
    The other one or two new edges are diagonals in newly introduced subgraphs isomorphic to $C_4,$ and so are added to enforce that $\partial_{k,k}^{r-1}(\bar{x}^i)=0$ and $\partial_{k,k}^{r+1}(\bar{x}^i)=0,$ as needed. Finally, by Lemma \ref{lem:diff_zero_so_edge}, the edge $\{x_{r-1}, x_{r+1}\}$ must be absent from $G$ to ensure that $\partial_{k,k}^r(\bar{x}^1) = \partial_{k,k}^r(\bar{x}^2) \neq 0$.
    Finally, from Lemma \ref{lem:diff_zero_so_edge} all other terms in the differential of both chains are zero due to edges in $E^{\{1\}}.$ 
\end{example}

    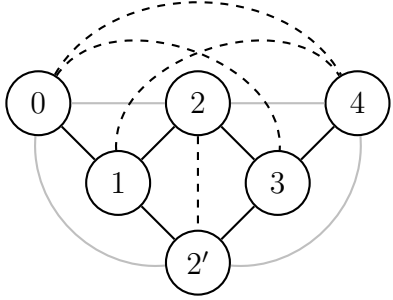
\begin{figure}[ht]
		\centering
        
		\begin{tikzpicture}[node distance={15mm}, thick, main/.style = {draw, circle,minimum size=.85cm}]
			\node[main] (1) {$0$};  
		\node[main] (2) [below right of=1] {$1$};
		\node[main] (3) [above right of=2] {$2$};
		\node[main] (4) [below right of=3] {$3$};
		\node[main] (5) [above right of=4] {$4$};
		\node[main] (6) [below right of=2] {$2'$};

		\draw (1) -- (2);
		\draw (2) -- (3);
		\draw (3) -- (4);
        \draw (4) -- (5);
			
		\draw[black] (4) -- (6);
		\draw[black] (2) -- (6);
		\draw[lightgray] (1) -- (3);
		\draw[lightgray] (3) -- (5);
		\draw[dashed] (3) -- (6);
			
		\path[every node/.style={font=\sffamily\small}]
			(1) edge [bend right=50, lightgray] node {} (6);
	    \path[every node/.style={font=\sffamily\small}]	
            (6) edge [bend right=50, lightgray] node {} (5);
        \path[every node/.style={font=\sffamily\small}]
			(2) edge [bend left=75, dashed] node {} (5);
	    \path[every node/.style={font=\sffamily\small}]
			(1) edge [bend left=75, dashed] node {} (4);
	    \path[every node/.style={font=\sffamily\small}]
			(1) edge [bend left=60, dashed] node {} (5);
			
		\end{tikzpicture} 
        \vspace*{1mm}
		\caption{Subgraph $H$ induced by the $(4,4)$-$EMH$ cycle $[0,1,2,3,4]-[0,1,2',3,4]$. In this case, the edge $(1,3)$ cannot be present in order to have $\partial_{4,4}(0,1,2,3,4) \neq 0$ and $\partial_{4,4}(0,1,2',3,4) \neq 0$, but $\partial_{4,4}((0,1,2,3,4)-(0,1,2',3,4)) = 0$. The black edges are in the support of the two paths, the grey edges needed to be added for the differential to vanish, and the dashed edges do not play a role in the homology computation.}
		\label{fig:subgraph emh m-combination}
	\end{figure}

It is possible to generalize the construction presented in Example \ref{ex:two_trail_cycle_graph} to the point where we remove all edges $\{x_{i-1},x_{i+1}\}$, and in this case the graph $H=(V_H,E_H)$ would be defined as
\begin{align*}
    V_H = &V^{\{1\}} \cup \{x_r' \colon r \in [k-1] \} \\
    E_H = &\left(E^{\{1\}} \cup \{\{x_{r-1}^1,x_r'\}, \{x_r',x_{r+1}^1\}: r \in [k-1] \} \cup \{x_{r}',x_{r+1}' \colon r \in [k-2]\} \right) \\
    & \setminus \{ \{x_{r-1}^1, x_{r+1}^1\} \colon r \in [k-1] \},
\end{align*}
and we see that $H$ is homomorphic to a grid graph of dimension $(k-2) \times (k-2)$. See Figure \ref{fig:subgraph emh max combination} for an illustration.

\begin{figure}[ht]
		\centering
        
		\begin{tikzpicture}[node distance={15mm}, thick, main/.style = {draw, circle,minimum size=.85cm}]
		\node[main] (0) {$0$};  
		\node[main] (1) [below right of=0] {$1$};
        \node[main] (1') [above right of=0] {$1'$};
        \node[main] (2) [above right of=1] {$2$};
        \node[main] (2') [below right of=1] {$2'$};
		\node[main] (3) [below right of=2] {$3$};
        \node[main] (3') [above right of=2] {$3'$};
		\node[main] (4) [above right of=3] {$4$};

		\draw (0) -- (1);
		\draw (1) -- (2);
		\draw (2) -- (3);
        \draw (3) -- (4);

        \draw (0) -- (1');
        \draw (1') -- (2);
        \draw (2) -- (3');
        \draw (3') -- (4);
        
		\draw (3) -- (2');
		\draw (1) -- (2');
        \draw (3') -- (2');
        \draw (1') -- (2');
			
		\end{tikzpicture} 
        \vspace*{1mm}
		\caption{Subgraph $H$ induced by the $(4,4)$-$EMH$ cycle $\bar{x}_1 - \bar{x}_2 + \bar{x}_3 - \bar{x}_4 + \bar{x}_5 - \bar{x}_6 + \bar{x}_7 - \bar{x}8$, where $\bar{x}_1=(0,1,2,3,4)$, $\bar{x}_2=(0,1',2,3,4)$, $\bar{x}_3=(0,1',2',3,4)$, $\bar{x}_4=(0,1',2',3',4)$, $\bar{x}_5=(0,1',2,3',4)$, $\bar{x}_6=(0,1,2,3',4)$, $\bar{x}_7=(0,1,2',3',4)$, $\bar{x}_8=(0,1,2',3,4)$. In this case, the edges $(0,2)$, $(1,3)$ and $(2,4)$ cannot be present in order to have $\partial_{4,4}\bar{x}_i \neq 0$ but $\partial_{4,4}(\sum_i (-1)^i \bar{x_i}) = 0$.}
		\label{fig:subgraph emh max combination}
	\end{figure}
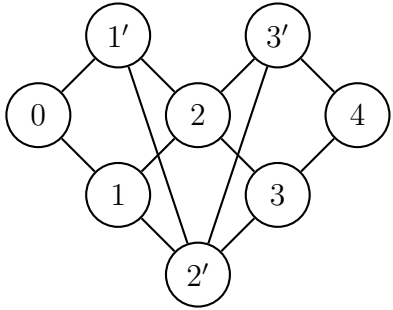

\begin{definition}
    \label{def:family of graphs induced by EMH cycles}
    Let $G=(V,E)$ be a graph.
    We call $\mathcal{H}=\{H_i\}$ the family of subgraphs of $G$ such that $H_i$ is induced by an $(k,k)$-eulerian magnitude homology cycle for every $i$.
    Notice that the minimal element of $\mathcal{H}$ is the graph in Figure \ref{fig:single-tuple-subgraph-emh}, while the maximal element is homomorphic to the grid graph of dimension $(k-2)\times (k-2)$ as in Figure \ref{fig:subgraph emh max combination}.
\end{definition}

The construction above implies the following result, which we will rely on to prove the complexity of computing eulerian magnitude homology.

\begin{theorem}
    \label{thm:cycle decomposition}
    Let $\{\bar{x}^i\}$ be a collection of tuples in $EMC_{k,k}(G)$.
    Then a linear combination $\sum_i a_i \bar{x}^i$ is an eulerian magnitude homology cycle, i.e. $\left[\sum_i a_i \bar{x}^i\right] \in EMH_{k,k}(G)$, if and only if the subgraph $H$ induced by the vertex set of $\sum_i a_i \bar{x}^i$ belongs to the family $\mathcal{H}$ defined in Definition \ref{def:family of graphs induced by EMH cycles}.
\end{theorem}

Theorem \ref{thm:cycle decomposition} (which we notice is a reformulation of \cite[Theorem 3.5]{giusti2024eulerian}) highlights the close relationship between computing the first diagonal of eulerian magnitude homology $EMH_{k,k}(G)$ and enumerating specific subgraphs $H_i$ of $G$.
More precisely, computing the rank of the $(k,k)$-eulerian magnitude homology group is equivalent to enumerating the subgraphs of $G$ isomorphic to graphs in the family $\mathcal{H}$ defined in Definition \ref{def:family of graphs induced by EMH cycles}.
In other words, our problem is equivalent to enumerating all isomorphisms from $\mathcal{H}$ to $G$, and in what follows we will use the connection to prove that computing $EMH_{k,k}(G)$ is $\#W[1]$-complete.  

\subsection{\texorpdfstring{$\#W[1]$}{W[1]}-completeness}

\emph{Parameterized complexity theory} provides a framework for a fine-grain complexity analysis of algorithmic problems that are intractable in general.
The core idea of the theory is \emph{fixed-parameter tractability}, which modifies the classical concept of polynomial time computability by allowing algorithms that run in exponential time, but only relative to a specific parameter of the problem that is typically small in practical applications.
A good example of this is database query evaluation: often, the query size $k$ is very small compared to the database size $n$.
Therefore, an algorithm that evaluates the query in $O(2^k\cdot n)$ time may be acceptable and even efficient.
In contrast, an algorithm with a $\Omega(n^k)$ evaluation time is generally not practical.
Fixed-parameter tractability hinges on this principle: a parameterized problem is considered fixed-parameter tractable if there exists a computable function $f$ and a constant $c$ such that the problem can be solved in $f(k)\cdot n^c$ time, where $n$ is the input size and $k$ is the parameter value.

To show that parameterized problems like the clique problem are not fixed-parameter tractable, a theory of \emph{parameterized intractability} has been developed~\cite{downey2012parameterized}.
This theory has led to the creation of a complex array of parameterized complexity classes.
Among these, the most significant are the $W[t]$ classes for $t \geq 1$, which together form the $W$-hierarchy.
It is believed that $W[1]$ contains the class FPT (fixed-parameter tractable problems) and that the $W$-hierarchy is strictly ordered.
Many natural parameterized problems belong to one of the classes within the $W$-hierarchy.
For instance, the parameterized clique problem is known to be complete for the class $W[1]$.

We recall now some concepts that are needed to define the $W$-hierarchy.

\begin{definition}
    A \emph{Boolean circuit} is a directed acyclic graph with the nodes labeled as follows:
    \begin{itemize}
        \item every node of in-degree $0$ is an input node
        \item every node with in-degree $1$ is a negation node ($\neg$)
        \item every node with in-degree $\geq 2$ is either an AND-node ($\land$) or an OR-node ($\lor$).
    \end{itemize}

Moreover, exactly one node with out-degree $0$ is also labeled the output node.
The \emph{depth} of the circuit is the maximum length of a directed path from an input node to the output node.
The \emph{weft} of the circuit is the maximum number of nodes with in-degree $\geq 3$ on a directed path from an input node to the output node.
\end{definition}

Given an assignment of Boolean values to the input gates, the circuit determines Boolean values at each node in the obvious way.
If the value of the output node is $1$ for an input assignment, we say that this assignment satisfies the circuit.
The weight of an assignment is its number of $1$s.

    \noindent\fbox{\parbox{\textwidth}{
    \begin{center}
        WEIGHTED CIRCUIT SATISFIABILITY (WCS)
        \\
    \end{center}
    
        Instance: A Boolean circuit $C$, an integer $k$
        
        Parameter: $k$

        Problem: Is there an assignment with weight $k$ that satisfies $C$?
        }}
    \\\\
It is known ~\cite{downey2012parameterized} that WCS has exponential complexity.

\begin{definition}
    The class of circuits $C_{t,d}$ contains the circuits with weft $\leq t$ and depth $\leq d$.
\end{definition}

For any class of circuits $\mathcal{C}$, we can define the following problem.

    \noindent\fbox{\parbox{\textwidth}{
    \begin{center}
        WCS($\mathcal{C}$)
        \\
    \end{center}
    
        Instance: A Boolean circuit $C \in \mathcal{C}$, an integer $k$
        
        Parameter: $k$

        Problem: Is there an assignment with weight $k$ that satisfies $C$?
        }}
    \\\\

\begin{definition}[$W$ and $\#W$-hierarchy]
Let $t \in \{1,2,\dots\}$.
A parameterized problem $\Pi$ is in the parameterized complexity class $W[t]$ if there exists a parameterized reduction from $\Pi$ to $WCS[C_{t,d}]$ for some constant $d\geq 1$.

Similarly, we define classes $\#W[t]$ for $t\geq 1$ is the class of all parameterized counting problems that are fixed-parameter reducible to $\#WCS[C_{t,d}]$ for some constant $d\geq 1$.
\end{definition}

\begin{theorem}[{\cite[Corollary 4]{chen2008understanding}}] 
\label{thm:StrEmb is W[1] complete}
Let $\mathcal{C}$ be a family of graphs.
Call StrEmb$(\mathcal{C})$ (strong embedding) the problem of asking whether a graph $C \in \mathcal{C}$ is isomorphic to an induced subgraph of a graph $B$, and call $\#-$StrEmb$(\mathcal{C})$ the counting analog. 
Then, if the graphs in $\mathcal{C}$ have bounded size, then StrEmb$(\mathcal{C}) \in \P$. Otherwise, $p-$StrEmb$(\mathcal{C})$ is complete for $W[1]$ under FPT many-one reductions.
The analogue holds for the counting problem $\#-$StrEmb$(\mathcal{C})$.
\end{theorem}

We thus have the following result.

\begin{theorem}
    Computing the rank of the $(k,k)$-eulerian magnitude homology group $EMH_{k,k}(G)$ is $\#W[1]$-complete under FPT many-one reductions.
\end{theorem}

\begin{proof}
    From Theorem \ref{thm:cycle decomposition} we know that computing the rank of the group $EMH_{k,k}(G)$ is equivalent to enumerating the isomorphisms from graphs in the family $\mathcal{H}$ defined in Definition \ref{def:family of graphs induced by EMH cycles} to induced subgraphs of $G$.
    We noticed that the graphs in the family $\mathcal{H}$ do not have bounded size, being the maximal element homomorphic to the grid graph of dimension $(k-2)\times (k-2)$.
    Therefore, using Theorem \ref{thm:StrEmb is W[1] complete} we can conclude that enumerating isomorphisms from graphs in $\mathcal{H}$ to induced subgraphs in $G$ is $\#W[1]$-complete under FPT many-one reductions, and thus so is the problem of computing the rank of the $(k,k)$-eulerian magnitude homology group.
\end{proof}
 
We determined that calculating EMH is a computationally challenging problem, and in the next section we propose an algorithm to deal with this task.

\section{First Diagonal Algorithm}
\label{sec:algorithm}
 
In this section we propose a method to compute eulerian magnitude chains $EMC_{k,k}(G)$ and $EMC_{k-1,k}(G)$, and the first diagonal eulerian magnitude homology groups of a graph $G$, $EMH_{k,k}(G)$.

\begin{algorithm}
    \begin{algorithmic}[1]
        \Function{FDA}{graph $G=(V,E)$, length of longest simple path $L$}
            \For{$v_0 \in V$}
                \State $V_0$ $\gets$ $V \setminus v_0$
                \For{$v_1$ adjacent to $v_0$}
                \label{for:start second loop}
                    \State $V_0$ $\gets$ $V_0 \setminus v_1$
                    \State $EMC_{1,1}(G)$.append$(v_0,v_1)$
                    \For{$v_2 \in V_0$ and $v_2$ adjacent to $v_1$}
                    \label{for:start third loop}
                        \State $V_0$ $\gets$ $V_0 \setminus v_2$
                        \State $EMC_{2,2}(G)$.append$(v_0,v_1,v_2)$
                        \If{$\len(v_0,v_2)=2$}
                            $EMC_{1,2}(G)$.append$(v_0,v_2)$ 
                        \EndIf
                        \For{$k$ in range$(2,L-1)$}
                        \label{for:start fourth loop}
                            \For{each $(k+1)$-tuple in $EMC_{k,k}(G)$}
                                \State $u$ $\gets$ $(k+1)$-tuple$[-1]$
                                \State $V_0$ $\gets$ $V_0 \setminus u$
                                \For{$w \in V_0$ and $w$ adjacent to $u$}
                                    \State $EMC_{k+1,k+1}(G)$.append($(k+1)$-tuple$+w$)
                                    \If{$EMC_{k-1,k}(G)$ is not empty}
                                        \For{$k$-tuple in $EMC_{k-1,k}(G)$}
                                            \State $EMC_{k,k+1}(G)$.append($k$-tuple$+w$)
                                        \EndFor
                                    \EndIf
                                    \If{distance between entry ($(k+1)$-tuple$+w$)[-3]\linespread{1.5} \algoIndent{7}
                                    and entry ($(k+1)$-tuple$+w$)[-1]) is $2$\algoIndent{7}}
                                        \State $EMC_{k,k+1}(G)$.append($(k+1)$-tuple$-u+w$)
                                    \EndIf
                                \EndFor
                            \EndFor
                        \EndFor
                        \label{for:end fourth loop}
                    \EndFor
                    \label{for:end third loop}
                \EndFor
                \label{for:end second loop}
            \EndFor
            
        \Return $EMC_{k-1,k}(G),EMC_{k,k}(G)$ for all $k \in [2,L]$\;
        \EndFunction
    \end{algorithmic}
    \caption{Algorithm to compute $EMC_{k,k}(G)$ and $EMC_{k-1,k}(G)$ for $2 \leq k \leq L$, where $L$ is the diameter of $G$.}
    \label{alg:alg chains LSP}
\end{algorithm}

The complexity of building the eulerian magnitude chains $EMC_{k,k}(G)$ highly depends on the density of the graph.
Indeed, to add the $i$-th vertex $v_i$ in the construction of a trail $(v_0,\dots,v_k) \in EMC_{k,k}(G)$ we are required by definition to verify that $v_i$ is different from the previous $i$ vertices, and this results in $k!$ checks for each trail in $EMC_{k,k}(G)$.

In order to overcome this problem we use a breadth-first-search-based approach.
Specifically, we choose a starting vertex $v_0$ and we build all possible $k$-paths starting at $v_0$ by performing BFS.
Then we repeat the procedure $n$ times (that is, we allow every vertex in $G$ to be the starting point $v_0$).
We call this algorithm First Diagonal Algorithm (FDA), and its pseudocode is shown in Algorithm \ref{alg:alg chains LSP}.

For what concerns eulerian magnitude homology groups $EMH_{k,k}(G)$, our approach is to build a sparse matrix with rows and column indexed with the elements of $EMC_{k,k}(G)$ and $EMC_{k-1,k}(G)$ respectively, and compute the kernel of such matrix.
The procedure is presented in Algorithm \ref{alg:alg homology}.

\begin{algorithm}
    \begin{algorithmic}[1]
        \Function{EMH}{$EMC_{k,k}$, $EMC_{k-1,k}$}

        \State RowVec $\gets$ [ ]
        \State ColVec $\gets$ [ ]
        \State Data $\gets$ [ ]

        \State index the columns with elements of $EMC_{k,k}$
        \For{ChainIndex $\in [\len(EMC_{k,k})-1]_0$}
            \State chain $\gets$ $EMC_{k,k}$[ChainIndex]
            \For{VtxIdx $\in [\len(\text{chain}) - 1]$}
                \If{removing any vertex does not change the length of a path\algoIndent{3}}
                    \If{the k-tuple with the vertex removed is part of $EMC_{k-1,k}$\algoIndent{4}}
                        \State RowVec.append($EMC_{k-1,k}$.index(chain), VtxIdx))))
                        \State ColVec.append(ChainIndex)
                        \State Data.append($-1^\text{VtxIdx}$)
                    \EndIf
                \EndIf
            \EndFor
        \EndFor

        \State ShapeMatrix $\gets$ ($\len(EMC_{k-1,k})$,$\len(EMC_{k,k})$))
        \State matrix $\gets$ SparseMatrix((Data, (RowVec, ColVec)), ShapeMatrix)
        \State betti $\gets$ dimension of kernel(matrix)

        \Return betti\;
        \EndFunction
    \end{algorithmic}
    \caption{Algorithm to compute Betti numbers $\beta_{k,k}$ for a chosen $k$.}
    \label{alg:alg homology}
\end{algorithm}

\subsection{Complexity of eulerian magnitude chain computation}
\label{subsec:algorithm complexity}

The complexity of the FDA presented Algorithm \ref{alg:alg chains LSP} highly depends on the connectivity of the graph $G$.
Indeed, apart for the first \emph{for loop} iterating on all the $n$ vertices, the other internal loops iterate on the neighbors of the considered vertex.

In our analysis we assume that the maximum degree of $G$ is $N_v$ (i.e. each vertex $v$ has at most $N_v$ neighbors).
Proceeding with this assumption it holds that
\sloppy
\begin{itemize}
    \item The second \emph{for loop}, lines \ref{for:start second loop}-\ref{for:end second loop}, (saving all edges $(v_0,v_1)$ in $EMC_{1,1}(G)$) performs at most $N_v$ iterations.
    \item The third \emph{for loop}, lines \ref{for:start third loop}-\ref{for:end third loop}, performs (at most) $(N_v)^2$ ``append'', $(N_v)^2$ ``if'' checks and $(N_v)^2$ more ``append'', for a total of $3(N_v)^2$ operations.
    \item The fourth \emph{for loop}, lines \ref{for:start fourth loop}-\ref{for:end fourth loop}, performs (at most) $(N_v)^3$ ``$EMC_{3,3}$.append'', $(N_v)^3$ ``$EMC_{2,3}$.append'', $(N_v)^3$ ``if'' checks and $(N_v)^3$ ``$EMC_{2,3}$.append'', for a total of (at most) $4(N_v)^3$ operations.
    \item In general, at the $m$-th step we perform at most $m \cdot (N_v)^{m-1}$ operations.
\end{itemize}

Therefore, an upper bound for the number of steps is $N_v + \sum_{i=3}^L i (N_v)^{i-1}$ where, we recall, $L$ is the diameter of the graph. 
Thus,
\[
N_v + \sum_{i=3}^L i (N_v)^{i-1} = N_v + \sum_{i=0}^L i (N_v)^{i-1} - 0 - 1 - 2N_v = (L-1)(N_v)^L - N_v,
\]

and we perform at most $n \cdot ((L-1)(N_v)^L - N_v) = \mathcal{O}(n (N_v)^L)$ operations, with $2\leq L \leq n$, where we achieve the lower bound $2$ if $G$ is a star graph and the upper bound $n$ if $G$ is a path graph.

\subsection{Discussion}
\label{subsec:discussion}

The FDA we just introduced does have exponential complexity in the worst-case scenario.
Indeed, the complexity depends on the diameter $L$ of the graph, which grows linearly with $n$ in the case, for example, of path graphs and cycle graphs.

Nevertheless, we present in this section some examples of graphs coming from real-world situations with a much smaller diameter, making FDA an effective tool for their analysis.

The concept of the ``small-world phenomenon'' describes a notable trend evident in various real-world graphs: the majority of vertex pairs are linked by paths significantly shorter than the overall size of the graph.
The \emph{diameter} of an undirected graph represents the longest shortest-path distance between any pair of vertices. It serves as a familiar measure indicating the ``small-world'' nature of the graph. In simpler terms, it gauges how swiftly one can traverse from one side of the graph to the opposite side.
The diameter is related to various processes, e.g. it is within a constant factor
of the memory complexity of the depth-first search algorithm. Also, it is a natural lower bound for the mixing time of any random walk~\cite{levin2017markov} and the broadcast time of the graph~\cite{hedetniemi1988survey}.
Mehrabian showed in 2017~\cite{mehrabian2017justifying} logarithmic upper bounds for the diameters of a variety of models, including the following well known ones: Erd\"{o}s-R\'{e}nyi random model~\cite{erdHos1960evolution}, forest fire model~\cite{leskovec2007graph}, copying model~\cite{kumar2000stochastic}, PageRank-based selection model~\cite{pandurangan2002using}, Aiello-Chung-Lu models~\cite{aiello2002random}, generalized linear preference model~\cite{bu2002distinguishing}, directed scale-free graphs~\cite{bollobas2003directed}, Cooper-Frieze model~\cite{cooper2003general}, and random unordered increasing k-trees~\cite{gao2009degree}.
This means that in each of these models, for every pair $(u, v)$ of vertices there exists a very short $(u, v)$-path connecting $u$ and $v$ whose length is
logarithmic in the number of vertices.

A similar result was proved in the 1990s concerning the structure of the World Wide Web.
Indeed, Albert, Jeong and Barab\'asi showed in~\cite{albert1999diameter} that the average shortest path $d$ between two documents (defined as the smallest number of URL links that must be followed to navigate from one document to the other) is $\langle d \rangle = 0.35 + 2.06 \log(N)$, where $N$ is the dimension of the network, indicating that the web forms a small-world network which characterizes social or biological systems.
Also, since for a given $N$, $d$ follows a gaussian distribution then $\langle d \rangle$ can be interpreted as the diameter of the web, a measure of the shortest distance between any two points in the system. 
Despite its huge size, estimated to be $N=8 \times 10^8$, this result indicates that the web is a highly connected graph with an average diameter of only $19$ links.

This implies, in particular, that the average diameters of many real-world models are logarithmic, which in turn signifies that in principle our algorithm FDA runs in $\mathcal{O}(n (N_v)^{\log(n)})$.
Although this means that FDA's complexity is in general super-polynomial, we point out two crucial facts:
\begin{enumerate}
    \item As noticed in Section \ref{sec:EMH computational cost}, computing EMH using the definition potentially results in exponential computational complexity.
    \item Many families of graphs have diameter $L$ much smaller than $\log(n)$. For example: complete graphs have $L=1$, star graphs have $L=2$, complete bipartite graphs have $L=2$, friendship graphs have $L=2$, $(k,\lambda,\mu)$-strongly regular graphs with $\mu > 0$ have $L=2$, and ``big enough" Cayley graph~\cite{erskine2018large} have diameter $L\leq 3$. 
\end{enumerate}

\sloppy
Therefore, even though FDA's worst case scenario is indeed super-polynomial, it does represent both a strong improvement of the definition and a feasible computational tool.

Consider the case of a real-world graph representing a Visual Social Network, \cite{boyd2007social}.
Over the past few years, there has been a tremendous surge in interest surrounding the analysis of Virtual Social Networks. This burgeoning field has attracted attention from a diverse array of disciplines including psychology, anthropology, sociology, economics, and statistics, transforming it into a truly interdisciplinary research domain.

By analyzing such networks, scientists aim to delve into the intricate web of relationships among individuals, groups, organizations, and other entities engaged in knowledge exchange across the digital landscape.
To achieve this, the following two methods are generally used \cite{d2010overview}:
\begin{enumerate}
    \item \emph{Socio-centric}: to examine sets of relationships between people that are regarded for analytical purposes as bounded social collectives, Figure \ref{fig:socio-centric}.
    \item \emph{Ego-centric}: to select focal individuals (egos), and identify the nodes they are connected to, Figure \ref{fig:ego-centric}.
\end{enumerate}

\begin{figure}[ht]
		\centering
        \begin{minipage}{0.45\linewidth}
        \hspace{1cm}
		\begin{tikzpicture}[node distance={15mm}, thick, main/.style = {draw, circle,minimum size=.75cm}]
            \node[main] (B) {$B$};
			\node[main] (A) [below of=B] {$A$};  
			\node[main] (C) [above right of=B] {$C$};
			\node[main] (D) [below right of=B] {$D$};
			\node[main] (E) [above right of=D] {$E$};
			\node[main] (F) [below of=E] {$F$};
			\node[main] (H) [below right of=E] {$H$};
            \node[main] (G) [above right of=H] {$G$};

			\draw (A) -- (B);
			\draw (B) -- (C);
			\draw (B) -- (D);
            \draw (A) -- (D);
			\draw (D) -- (E);
			\draw (C) -- (E);
            \draw (F) -- (E);
            \draw (E) -- (G);
            \draw (E) -- (H);
            \draw (G) -- (H);
            \path[every node/.style={font=\sffamily\small}]
			(D) edge [bend right=90] node {} (G);
			
		\end{tikzpicture} 
        \vspace*{-2mm}
		\caption{Socio-centric social network with individuals A, B, C, D, E, F, G, H. The diameter of this graph is $L=3$.}
		\label{fig:socio-centric}
        \end{minipage}
        \hfill
        \begin{minipage}{0.45\linewidth}
		\centering
		\begin{tikzpicture}[node distance={15mm}, thick, main/.style = {draw, circle,minimum size=.75cm}]
		\def \n {20}
        \def \N {8}
        \def \radius {2cm}
        \def \rd {1mm}
        \def \rer {4mm}

        \def \margin {8} 

        \node[main] at (360:0mm) (A) {$A$};
        \foreach \i [count=\ni from 0] in {B,C,D,E,F}{
        \node[main] at ({108-\ni*30}:\radius) (u\ni) {\i};
        \draw (A)--(u\ni);}

        \foreach \i in {3,...,11}{
            \node[circle] at ({-\i*20}:\radius) (aux) {\phantom{$G$}};
            \draw[dotted, shorten >=2mm, shorten <=2mm] (A)--(aux);}
		\end{tikzpicture} 
        \vspace*{-7mm}
		\caption{Ego-centric social network with individuals A, B, C, D, E, F, G. The diameter of this graph is $L=2$.}
		\label{fig:ego-centric}
        \end{minipage}
	\end{figure}

In the first case, one of the most important measures characterizing the network is the
density (or connectedness), which is the number of links in a network as a ratio of
the total possible links.
As the density grows, the diameter of the graph describing the relationships in the visual network will decrease, making the computation of eulerian magnitude homology more feasible.

If the second method is used, then we are focusing on the relations starting from one specific individual, which will produce a star graph.
Therefore, also in this case eulerian magnitude homology represents a viable tool.

\section{Conclusions and perspectives}

In this paper we investigated the computational cost of calculating the ranks of first diagonal eulerian magnitude homology groups of a graph $G$, $EMH_{k,k}(G)$.
We first infer that this problem is $\#W[1]$-complete by rewriting it as a subgraph isomorphism problem. 
Then, we produce FDA, a BFS-based algorithm which computes the eulerian magnitude chains $EMC_{k,k}(G)$ and $EMC_{k-1,k}(G)$ in super-polynomial time.

\medspace

The tools developed in this paper suggest a number of avenues for further work.
Here we highlight the ones that the authors find particularly interesting.

\begin{enumerate}
    \item In order to base our algorithm on a breadth-first-search, we heavily leverage equality between the length $\ell$ and number of landmarks $k$. We believe that this method can be leveraged to iteratively study the subgraphs inherent in the $\ell = k+i$ lines for increasing $i$, providing more insight into the graph substructures.
    \item The one presented here is the first algorithm computing eulerian magnitude homology in non-exponential time, thus representing a high improvement of definition's computational complexity. Nevertheless, the authors believe the present algorithm's efficiency can be further improved with techniques similar to the ones used in the case of Persistent Homology~\cite{otter2017roadmap}. This would initiate a new current in the field of Topological Data Analysis based on magnitude homology and devoted to network analysis. 
\end{enumerate}

\subsection*{Acknowledgments}
The authors are thankful to Giulia Bernardini for useful discussions throughout the development of this work. 

\bibliographystyle{amsplain}
\bibliography{sn-bibliography.bib}

\end{document}